%% file: paper.tex
\documentclass[11pt,pdftex,letterpaper]{article}
\pdfoutput=1
\usepackage{floatrow}
\newfloatcommand{capbtabbox}{table}[][\FBwidth]
\usepackage{amssymb}
\usepackage{cite}
\usepackage{multicol}
\usepackage{amsmath}
\usepackage{floatflt}
\usepackage{flushend}
\usepackage[margin=0.9in]{geometry}

\ifx\pdftexversion\undefined
  \usepackage[dvips]{graphicx}
\else
  \usepackage[pdftex]{graphicx}
\fi

\PassOptionsToPackage{pdftex}{graphicx}

\usepackage{tikz}
\usepackage{macros}
\usepackage[plain]{algorithm}
\usepackage{algpseudocode}
\usepackage{ioa_code}

\newtheorem{theorem}{Theorem}

\newtheorem{corollary}[theorem]{Corollary}
\newtheorem{definition}[theorem]{Definition}
\newtheorem{lemma}[theorem]{Lemma}

\newenvironment{proof}[1][Proof]{\noindent\textbf{#1.} }{\hfill $\Box$\\[2mm]}
\newenvironment{proofsketch}[1][Proof sketch]{\noindent\textbf{#1.} }{\hfill $\Box$\\[2mm]}

\newcounter{linenumber}

\newcounter{excount}
{
\stepcounter{excount} {\it Example \arabic{excount} (#1)} 
}
{
}

\def\S{\ensuremath{\mathcal{S}}}

\def\N{\ensuremath{\mathbb{N}}}

\def\Nat{\ensuremath{\mathbb{N}}}

\newcommand{\LS}{LS}

\newcommand{\true}{\lit{true}}
\newcommand{\false}{\lit{false}}

\newcommand{\remove}[1]{}

\newcommand{\id}[1]{\mbox{\textit{#1}}}

\newcommand{\petr}[1]{{\bf (((PK: #1)))}}

\newcommand{\LL}{\ms{LL}}
\newcommand{\ignore}[1]{}

\begin{document}

\title{Optimism for Boosting Concurrency}
\author{
Vincent Gramoli\,$^1$~~~Petr Kuznetsov\,$^{2}$~~~Srivatsan Ravi\,$^3$\thanks{Contact author:
          srivatsan@srivatsan.in, FG INET, MAR 4-4, Marchstr. 23,
          10587 Berlin, Germany}\\
\small $^1$ NICTA and University of Sydney\\
\small $^2$ T\'el\'ecom ParisTech\\
\small $^3$ TU Berlin/ Purdue University
}

\date{}\maketitle
\thispagestyle{empty}

\begin{abstract}
Modern concurrent programming benefits from a large variety of synchronization techniques.
These include conventional pessimistic locking, as well as optimistic
techniques based on conditional synchronization primitives or 
transactional memory. Yet, it is unclear which of these approaches 
better leverage the concurrency inherent to multi-cores.

In this paper, we compare the level of concurrency one can obtain by converting a 
sequential program into a concurrent one using optimistic or
pessimistic techniques. 
To establish fair comparison of such implementations, we 
introduce a new correctness
criterion for concurrent programs, defined independently 
of the synchronization techniques they use. 

We treat a program's concurrency as its ability to accept a
concurrent schedule, a metric inspired by the theories of both databases
and transactional memory. We show that pessimistic locking can
provide strictly higher concurrency than transactions for some
applications whereas transactions can provide strictly higher
concurrency than pessimistic locks for others. 
Finally, we show that combining the benefits of the two synchronization techniques 
can provide  strictly more concurrency than any of them individually.
We propose a list-based set algorithm that is optimal in the sense that it accepts all
correct concurrent schedules.
As we show via experimentation, the optimality in terms of concurrency is
reflected by scalability gains.      
\end{abstract}
%

\newpage
\pagenumbering{arabic}\setcounter{page}{1}
%
\input{intro}
\section{Preliminaries}
\label{sec:prel}
%
\vspace{1mm}\noindent\textbf{Sequential types and implementations.}
%
An \emph{object type} $\tau$ is a tuple
$(\Phi,\Gamma, Q, q_0, \delta)$ where
$\Phi$ is a set of operations,
$\Gamma$ is a set of responses, $Q$ is a set of states, $q_0\in Q$ is an
initial state and 
$\delta \subseteq Q\times \Phi \times Q\times \Gamma$ 
is a transition relation that determines, for each state
and each operation, the set of possible
resulting states and produced responses~\cite{AFHHT07}. 
For any type $\tau$, each high-level object $O_{\tau}$ of this type has a \emph{sequential implementation}. 
For each operation $\pi \in \Phi$, 
$\id{IS}$ specifies a deterministic procedure that 
performs \emph{reads} and \emph{writes} on a collection of objects
$X_1,\ldots , X_m$ that encode a state of $O_{\tau}$, and returns a response $r\in \Gamma$. 

As a running example, we consider the sorted linked-list based implementation of the type \emph{set}, commonly referred to 
as the \emph{list-based set}~\cite{HS08-book}.
The \emph{set} type exports operations $\lit{insert}(v)$, $\lit{remove}(v)$ and
$\lit{contains}(v)$, with $v\in\mathbb{Z}$.
We consider a sequential implementation $\LL$
of the \emph{set} type using a sorted linked list where 
each element (or \emph{object}) stores an integer value, $\ms{val}$, and a pointer to its successor, $\ms{next}$, so that elements are 
sorted in the ascending order of their value.  
Both element fields are accessed atomically.
Every operation invoked with a parameter $v$ traverses the list starting from the
$\ms{head}$ up to the element storing value $v'\geq v$.
If $v'=v$, then $\lit{contains}(v)$ returns $\lit{true}$, $\lit{remove}(v)$ unlinks the 
corresponding element and returns $\lit{true}$, and $\lit{insert}(v)$ returns $\lit{false}$. Otherwise, 
$\lit{contains}(v)$ and $\lit{remove}(v)$ return
$\lit{false}$ while $\lit{insert}(v)$ adds a new element with value
$v$ to the list and returns $\lit{true}$. 
The list-based set 
is denoted by $(\LL,\ms{set})$ (cf. formal definition in Appendix~\ref{app:seq}).

\vspace{1mm}\noindent\textbf{Concurrent implementations.}
We tackle the problem of turning the sequential
implementation $\id{IS}$ of type $\tau$ into a \emph{concurrent} one, shared by 
$n$ \emph{processes} $p_1,\ldots,p_n$ ($n\in\Nat$).
The implementation provides the processes with algorithms 
for the reads and writes on objects.
We refer to the resulting implementation as a concurrent implementation of $(\id{IS},\tau)$.
We assume an asynchronous shared-memory system in which the processes communicate by
applying primitives on shared \emph{base objects}~\cite{Her91}.
We place no upper bounds on the number of versions an object may maintain or on the size of this object.
Throughout this paper, the term \emph{operation} refers to some
high-level operation of the type, 
while read-write operations on objects are referred simply 
as \emph{reads} and \emph{writes}.

An implemented read or write may \emph{abort} by returning a special response
$\bot$. In this case we say that the corresponding high-level
operation is \emph{aborted}. 
The $\bot$ event is treated both as the response event of the read or
write operation and as the response of the corresponding high-level operation.   

\vspace{1mm}\noindent\textbf{Executions and histories.}
An \emph{execution} of a concurrent implementation is a sequence
of invocations and responses of high-level operations of type $\tau$, 
invocations and responses of read and write
operations, and invocations and responses of base-object primitives.
We assume that executions are \emph{well-formed}:
no process invokes a new read or write, or high-level operation before
the previous read or write, or a high-level operation, resp., 
returns, or takes steps outside its read or write operation's interval.

Let $\alpha|p_i$ denote the subsequence of an execution $\alpha$
restricted to the events of process $p_i$.
Executions $\alpha$ and $\alpha'$ are \emph{equivalent} if for every process
$p_i$, $\alpha|p_i=\alpha'|p_i$.
An operation $\pi$ \emph{precedes} another operation $\pi'$ in an execution
$\alpha$, 
denoted $\pi \rightarrow_{\alpha} \pi'$, 
if the response of $\pi$ occurs before the invocation of $\pi'$.
Two operations are \emph{concurrent} if neither precedes
the other. 
An execution is \emph{sequential} if it has no concurrent 
operations. 
A sequential execution $\alpha$ is \emph{legal} 
if for every object $X$, every read of $X$ in $\alpha$ 
returns the latest written value of $X$.
An operation is \emph{complete} in $\alpha$ if the invocation event is
followed by a \emph{matching} (non-$\bot$) response or aborted; otherwise, it is \emph{incomplete} in $\alpha$.
Execution $\alpha$ is \emph{complete} if every operation is complete in $\alpha$.

The \emph{history exported by an execution $\alpha$} is
the subsequence of $\alpha$ reduced to the invocations and responses
of operations,  reads and writes, except for the reads
and writes that return $\bot$. 

\vspace{1mm}\noindent\textbf{High-level histories and linearizability.}
%
A \emph{high-level history} $\tilde H$ of an execution $\alpha$ is the subsequence of $\alpha$ consisting of all
invocations and responses of \emph{non-aborted} operations.
A complete high-level history $\tilde H$ is \emph{linearizable} with 
respect to an object type $\tau$ if there exists
a sequential high-level history $S$ equivalent to $H$ such that
(1) $\rightarrow_{\tilde H}\subseteq \rightarrow_S$ and
(2) $S$ 
is consistent with the sequential specification of type $\tau$.
Now a high-level history $\tilde H$ is linearizable if it can be
\emph{completed} (by adding matching responses to a subset of
incomplete operations in $\tilde H$ and removing the rest)
to a linearizable high-level history~\cite{HW90,AW04}.

\vspace{1mm}\noindent\textbf{Obedient implementations.}
We only consider implementations that satisfy the following condition:
Let $\alpha$ be any complete sequential execution of a concurrent implementation $I$.
Then in every execution of $I$ of the form $\alpha\cdot\rho_1\cdots \rho_k$
where each $\rho_i$ ($i=1,\ldots,k$) is the complete execution of a
read, every read returns the value written by the last write that does
not belong to an aborted operation.

Intuitively, this assumption restricts our scope to
``obedient'' implementations of reads and writes, where no
read value may depend on some future write.   
In particular, we filter out implementations in which the
complete execution of a high-level operation is performed within the
first read or write of its sequential algorithm.

\vspace{1mm}\noindent\textbf{Pessimistic implementations.}
Informally, a concurrent implementation is \emph{pessimistic} if the exported history contains
every read-write event that appears in the execution. 
More precisely, no execution of a pessimistic implementation includes
operations that returned $\bot$.  

For example, a class of pessimistic implementations are those based on \emph{locks}.
A lock provides
shared or exclusive access to an object $X$ through 
synchronization primitives $\lit{lock}^S(X)$ (\emph{shared mode}),
$\lit{lock}(X)$ (\emph{exclusive mode}),  
and $\lit{unlock}(X)$.
When $\lit{lock}^S(X)$ (resp. $\lit{lock}(X)$) invoked
by a process $p_i$ returns, we say that $p_i$ \emph{holds
a lock on $X$ in shared (resp. exclusive) mode}.
A process \emph{releases} the object it holds by invoking
$\lit{unlock}(X)$.  
If no process holds a shared or exclusive
lock on $X$, then $\lit{lock}(X)$
eventually returns;
if no process holds an exclusive
lock on $X$, then $\lit{lock}^S(X)$
eventually returns; and
if no process holds a
lock on $X$ forever, then every $\lit{lock}(X)$ or $\lit{lock}^S(X)$
eventually returns. 
Given a sequential implementation of a data type, 
a corresponding lock-based concurrent one 
is derived by inserting the synchronization primitives
to provide read-write access to an object. 

\vspace{1mm}\noindent\textbf{Optimistic implementations.}
In contrast with pessimistic ones, optimistic implementations may, under
certain conditions, abort an operation:
some read or write may return $\bot$,
in which case the corresponding operation also returns $\bot$.

Popular classes of optimistic implementations are those based on
``lazy synchronization''~\cite{HHL+05,HS08-book} (with the ability of
returning $\bot$ and re-invoking an operation) or   
 \emph{transactional memory} (\emph{TM})~\cite{HM93,ST95,HLR10}.
A TM provides access to a
collection of objects via \emph{transactions}.
A transaction is a sequence of read and write operations on
objects. A transaction may \emph{commit},
or one of the read or write performed by the transaction may \emph{abort}.
Given a sequential implementation of a data type, 
a corresponding TM-based concurrent one 
puts each sequential operation within a
transaction and replaces each read and write of an object $X$ with the
transactional read and write implementations, respectively. 
If the transaction commits, then the result of the operation is
returned to the user; otherwise if one of the transactional operations aborts, $\bot$ is returned.

\section{Locally serializable linearizability}
\label{sec:lin}
We are now ready to define the correctness criterion that we impose on our
concurrent implementations.

Let $H$ be a history and let $\pi$ be a high-level operation in $H$. 
Then $H|\pi$ denotes the subsequence of $H$ consisting of the events
of $\pi$, except for the last aborted read or write, if any.
Let $\id{IS}$ be a sequential implementation of an object of type
$\tau$ and $\Sigma_{\id{IS}}$, the set of histories of $\id{IS}$. 
\begin{definition}[LS-linearizability]
\label{def:lin}
A history ${H}$ is \emph{locally serializable with respect to}
${\id{IS}}$ if for every high-level operation $\pi$ in $H$,
there exists $S \in \Sigma_{\id{IS}}$ such that $H|\pi=S|\pi$.
A history ${H}$ is \emph{\LS-linearizable with respect to
$(\id{IS},\tau)$}  (we also write $H$ is $(\id{IS},\tau)$-LSL)  if:
(1) ${H}$ is locally serializable with respect to
$\id{IS}$ and (2) the corresponding high-level history $\tilde H$ 
is linearizable with respect to $\tau$.
\end{definition}
Observe that local serializability stipulates that the execution is 
witnessed sequential by every operation.
Two different operations (even when invoked by the same process) are not
required to witness mutually consistent sequential executions.

A concurrent implementation $I$ is \emph{\LS-linearizable with respect to
$(\id{IS},\tau)$} (we also write $I$ is $(\id{IS},\tau)$-LSL)
if every history exported by $I$ is $(\id{IS},\tau)$-LSL.   
Throughout this paper, when we refer to a concurrent implementation of $(\id{IS},\tau)$, 
we assume that it is \LS-linearizable with respect to $(\id{IS},\tau)$.

Just as linearizability, \LS-linearizability is
\emph{compositional}~\cite{HW90,HS08-book}: a composition of LSL 
implementations is also LSL. (cf. Appendix~\ref{app:comp}).
However, it is not \emph{non-nonblocking}: local
serializability may prevent an operation in a finite LSL
history from completing in a non-blocking manner. 

\begin{figure*}[t]
 \includegraphics[scale=0.52]{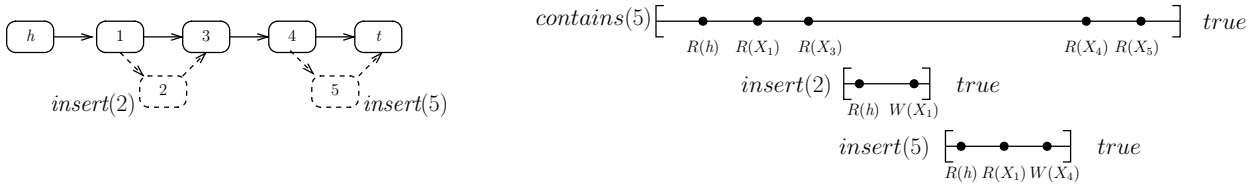}
 \caption{\small{A concurrency scenario for a list-based set, initially $\{1,3,4\}$, where value $i$ is stored at node $X_i$:
   $\lit{insert}(2)$ and  $\lit{insert}(5)$ can proceed
   concurrently with $\lit{contains}(5)$, the history is
   LS-linearizable but not serializable. (We only depict important read-write
   events here.)}}\label{fig:ex1}%
\end{figure*}

\vspace{1mm}\noindent\textbf{LS-linearizability versus other criteria.}
LS-linearizability is a two-level consistency criterion which makes it
suitable to compare concurrent implementations of a sequential data
structure, regardless of synchronization techniques they use.
It is quite distinct from related criteria designed for database and software
transactions, such as serializability~\cite{Pap79-serial,WV02-book} and
multilevel serializability~\cite{Wei86,WV02-book}.

For example, serializability~\cite{Pap79-serial} prevents sequences of reads and writes from conflicting in a cyclic way, 
establishing a  global order of transactions.
Reasoning only at the level of reads and writes may be overly conservative:
higher-level operations may commute even if their reads and writes conflict~\cite{Wei88}.
Consider an execution of a concurrent \emph{list-based set} depicted in 
Figure~\ref{fig:ex1}.
We assume here that the set initial state is $\{1,3,4\}$.
Operation $\lit{contains}(5)$ is concurrent, first with 
operation $\lit{insert}(2)$ and then with operation $\textsf{insert}(5)$. 
The history is not serializable:
$\lit{insert}(5)$ sees 
the effect of $\lit{insert}(2)$ because $R(X_1)$ by $\lit{insert}(5)$ returns the value
of $X_1$ that is updated by $\lit{insert}(2)$ and
thus should be serialized after it. But $\lit{contains}(5)$ misses
element $2$ in the linked list, but must see the
effect of $\lit{insert}(5)$ to perform the read of $X_5$, i.e., the element created by $\lit{insert}(5)$.  
However, this history is LSL since each of the three local histories is consistent with some
sequential history of $\LL$. 

Multilevel serializability~\cite{Wei86,WV02-book} was 
proposed to reason in terms of multiple semantic levels in the same execution.
\LS-linearizability, being defined for two levels only, does not require a global serialization of low-level operations as
$2$-level serializability does. 
LS-linearizability simply requires each process  to observe a local serialization, which can be different from one
process to another. Also, to make it more suitable for concurrency
analysis of a concrete data structure, instead of semantic-based commutativity~\cite{Wei88}, we use the sequential
specification of the high-level behavior of the object~\cite{HW90}.

Linearizability~\cite{HW90,AW04} only accounts for high-level
behavior of a data structure,  so it does not imply
LS-linearizability. For example, Herlihy's universal
construction~\cite{Her91} provides a linearizable implementation for
any given object type, but does not guarantee that each execution locally appears
sequential with respect to any sequential implementation of the type.    
Local serializability, by itself, does not require any synchronization
between processes and can be trivially implemented without
communication among the processes.
Therefore, the two parts of LS-linearizability indeed complement each other.  

%

%
\section{The concurrency metric}\label{sec:concurrency}
To characterize the ability of a concurrent implementation to process arbitrary interleavings of sequential code, we introduce 
the notion of a \emph{schedule}.
Intuitively, a schedule describes the order in which complete high-level
operations, and sequential reads and writes are invoked by the user. 
More precisely, a schedule is 
an equivalence class of complete histories that agree on
the \emph{order} of invocation and response events of reads, writes and high-level operations, but 
not necessarily on read \emph{values} or high-level responses.
Thus, a schedule can be treated as a history, where responses of reads and operations
are not specified. 

We say that an implementation $I$ \emph{accepts} a schedule $\sigma$ if 
it exports a history $H$ such that $\ms{complete}(H)$ exhibits
the order of $\sigma$, where $\ms{complete}(H)$ is the subsequence of $H$
that consists of the events of the complete operations that returned a matching response. 
We then say that the execution (or history) \emph{exports} $\sigma$. 
A schedule $\sigma$ is 
$(\ms{IS},\tau)$-LSL if there
exists an $(\id{IS},\tau)$-LSL history that exports $\sigma$.

A \emph{synchronization technique} is a set of concurrent implementations.
We define below a specific optimistic synchronization technique and then
a specific pessimistic one.

\vspace{1mm}\noindent\textbf{The class $\mathcal{SM}$.}
Let $\alpha$ denote the execution of a TM implementation and
$\ms{ops}(\alpha)$, 
the set of transactions each of which performs at least one event in $\alpha$.
Let ${\alpha}^k$ denote the prefix of $\alpha$ up to the last event of transaction $\pi_k$.
Let $\ms{Cseq}(\alpha)$ denote the set of subsequences of ${\alpha}$  that
consist of all the events of transactions that are committed and some
transactions that started committing in $\alpha$. 
We say that $\alpha$ is \emph{strictly serializable} if 
there exists a legal sequential execution $\alpha'$ equivalent to
a sequence in $\sigma\in\ms{Cseq}(\alpha)$
such that $\rightarrow_{\sigma} \subseteq \rightarrow_{\alpha'}$. 

This paper focuses on TM-based implementations that are strictly
serializable and, in addition, guarantee that every
transaction (even aborted or incomplete) observes correct (serial)
behavior.  
More precisely, an execution $\alpha$ is  \emph{safe-strict serializable} if
(1) $\alpha$ is strictly serializable, and
(2) for each operation $\pi_k$ that is incomplete or returned $\bot$ in
$\alpha$, there exist a legal sequential execution of transactions
$\alpha'=\pi_0\cdots \pi_i\cdot \pi_k$ and
$\sigma\in\ms{Cseq}(\alpha^k)$ such that $\{\pi_0,\cdots, \pi_i\}
\subseteq \ms{ops}(\sigma)$ and $\forall \pi_m\in \ms{ops}(\alpha'):{\alpha'}|m={\alpha^k}|m$.

Safe-strict serializability captures nicely both local serializability
and linearizability. 
If we transform a sequential implementation
$\id{IS}$ of a type $\tau$ into a concurrent one using any
\emph{safe-strict serializable} TM, 
we obtain an LSL TM-based
implementation of $(\id{IS},\tau)$. 
Indeed, by running each operation of $\id{IS}$ within a transaction of
a safe-strict serializable TM, we make sure that operations in 
committed transactions witness the same execution of $\id{IS}$, and
every operation that returned $\bot$ is consistent with some
execution of $\id{IS}$ based on previously completed operations.
Formally, $\mathcal{SM}$ denotes the set of TM-based LSL
implementations.
(We discuss the relations to similar but stronger TM criteria, such as
opacity~\cite{tm-book}, \emph{TMS1}~\cite{TMS09} and
\emph{VWC}~\cite{damien-vw} in Section~\ref{sec:related}.)   

\vspace{1mm}\noindent\textbf{The class $\mathcal{P}$.}
This denotes the set of \emph{deadlock-free} pessimistic
LSL implementations: assuming that every process takes
enough steps, 
at least one of the concurrent operations return a matching response~\cite{HS11-progress}.
Note that $\mathcal{P}$ includes implementations that are not necessarily safe-strict serializable. 
In the next section, we describe a pessimistic implementation of the list-based set that accepts non-serializable schedules by
fine-tuning to the semantics of the \emph{set} type.
%
\begin{figure*}
 \includegraphics[scale=0.45]{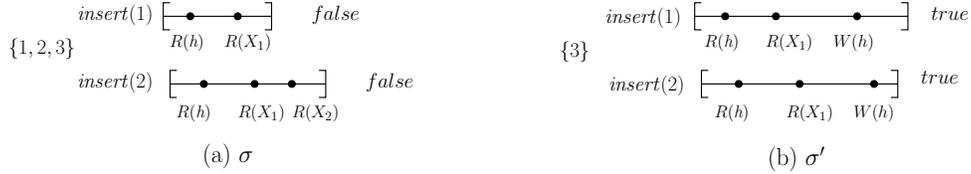}
 \caption{\small{%
(a) a history exporting schedule $\sigma$, with initial state
   $\{1,2,3\}$, accepted by $I^{C}\in \mathcal{SM}$; 
(b) a history exporting a problematic schedule $\sigma'$, with initial state 
   $\{3\}$, which should be accepted by any $I\in\mathcal{P}$ if it accepts $\sigma$}}\label{fig:ex2}%
\vspace{-0.35mm}
\end{figure*}
\section{On the incomparability of synchronization techniques}\label{sec:incomparability}
We now provide a concurrency analysis of synchronization techniques $\mathcal{SM}$ and $\mathcal{P}$
in the context of the list-based set.
%
We describe a pessimistic implementation of $(\LL,\ms{set})$, $I^H \in \mathcal{P}$,
that accepts non-serializable schedules: each read operation performed by $\lit{contains}$ 
acquires the \emph{shared lock} on the object, 
reads the $\ms{next}$ field of the element before releasing the shared lock on the predecessor element 
in a \emph{hand-over-hand} manner~\cite{BS88}.
Update operations ($\lit{insert}$ and
$\lit{remove}$) acquire the \emph{exclusive lock} on the $\ms{head}$ during $\lit{read}(\ms{head})$ 
and release it at the end. Every other read operation performed 
by update operations simply reads the element $\ms{next}$ field to traverse the list. The write operation
performed by an $\lit{insert}$ or a $\lit{remove}$ acquires the exclusive lock, writes the value
to the element and releases the lock.
There is no real concurrency between any two update operations since the process holds the 
exclusive lock on the $\ms{head}$ throughout the operation execution.
Note that $I^H$ is deadlock-free and $(\LL,\ms{set})$-LSL.

On the one hand, the schedule of $(\LL,\ms{set})$ depicted in Figure~\ref{fig:ex1}, which we denote by $\sigma_0$, is not serializable
as explained in Section~\ref{sec:lin}
and must be rejected by any implementation in $\mathcal{SM}$.
However, there exists an execution of $I^H$ that exports $\sigma_0$
since there is no read-write conflict on any two consecutive elements accessed.

On the other hand, consider the schedule $\sigma$ of $(\LL,\ms{set})$ in Figure~\ref{fig:ex2}(a).
Clearly, $\sigma$ is serializable and is accepted by most (progressive~\cite{tm-theory}) TM-based implementations since there
is no read-write conflict.
However, we prove that $\sigma$ is not accepted by any implementation in $\mathcal{P}$.
Our proof technique is interesting in its own right: we show that
if there exists any implementation in $\mathcal{P}$ that accepts $\sigma$, it must also
accept the schedule $\sigma'$ depicted in Figure~\ref{fig:ex2}(b). In $\sigma'$, 
$\lit{insert}(2)$ overwrites the write on \emph{head} performed by $\lit{insert}(1)$
resulting in a lost update. By deadlock-freedom, there exists an extension of $\sigma'$ in which
a $\lit{contains}(1)$ returns $\false$; but this is not a linearizable schedule. 
%
%
\begin{theorem}[Incomparability]
\label{th:mpl}
There exist schedules $\sigma_0$ and $\sigma$ 
of $(\LL,\ms{set})$ such that
(1) $\sigma_0$ is accepted by an $(\LL,\ms{set})$-LSL implementation $I^H \in \mathcal{P}$
but not accepted by \emph{any} $(\LL,\ms{set})$-LSL implementation in $\mathcal{SM}$, and
(2) $\sigma$ is accepted by an $(\LL,\ms{set})$-LSL implementation $I^C \in \mathcal{SM}$
but not accepted by \emph{any} $(\LL,\ms{set})$-LSL implementation in
$\mathcal{P}$.
(The proof is in Appendix~\ref{app:smp}.)
\end{theorem}
The second part of
Theorem~\ref{th:mpl} 
may look surprising, as 
the class $\mathcal{P}$ includes implementations that are relaxed (not safe-strict serializable) and 
fine-tuned to the semantics of the type whereas implementations in the class $\mathcal{SM}$ are oblivious to the semantics of the data type.
However, since TM-based implementations are optimistic, i.e., every read-write operation remains tentative, 
the implementation does not need to be overly conservative and could return $\bot$ in case a matching response to the 
operation cannot be returned.
\section{On the benefits of being optimistic and relaxed }
\label{sec:svlock}
We now combine the benefits of relaxation and optimism to derive an optimistic implementation 
of the list-based set
that supersedes every implementation in classes $\mathcal{P}$ and
$\mathcal{SM}$ in terms of concurrency.
Our implementation, denoted $I^{RM}$ provides processes with algorithms for 
implementing read and write operations on the elements of the list for each operation of the 
list-based set (Algorithm~\ref{alg:elastic}).
\input{simple-elastic-s3}
%
%

Every object (or element) $X_{\ell}$ is
specified 
by the following shared variables: $\ms{t-var}[\ell]$
stores the \emph{value} $v\in V$ of $X_{\ell}$, $r[\ell]$ stores
a boolean indicating if $X_{\ell}$ is \emph{marked for deletion}, 
$L[\ell]$ stores a tuple of the \emph{version number} of $X_{\ell}$ and a \emph{locked} flag; 
the latter indicates whether a concurrent process is performing a write to $X_{\ell}$.

Any operation with input parameter $v$ traverses the list starting from the
$\ms{head}$ element up to the element storing value $v'\geq v$ without writing to shared memory.
If a read operation on an element conflicts with a write operation to the same element or
if the element is marked for deletion, the operation terminates by returning $\bot$.
While traversing the list, the process maintains the last two read elements and their version numbers 
in the local rotating buffer $\ms{rbuf}$. If none of the read operations performed by $\lit{contains}(v)$ return $\bot$
and if $v'=v$, then $\lit{contains}(v)$ returns $\lit{true}$; otherwise it returns $\lit{false}$.
Thus, the $\lit{contains}$ does not write to shared memory.

To perform write operation to an element as part of an update operation ($\lit{insert}$ and $\lit{remove}$), the process 
first retrieves the version of the object that belongs to its rotating buffer.
It returns $\bot$ if the version has been changed since the previous read of the element 
or if a concurrent process is executing a write to the same element.
Note that, technically, $\bot$ is returned only if $\ms{prev.next}\not\rightarrow \ms{curr}$.
If $\ms{prev.next}\rightarrow \ms{curr}$, then we attempt to lock the element with the current version
and return $\bot$ if there is a concurrent process executing a write to the same element.
But we avoid expanding
on this step in our algorithm pseudocode.
The write operation performed by the $\lit{remove}$ operation, additionally checks if the element to be removed
from the list is locked by another process; if not, it sets a flag on the element to mark it for deletion.
If none of the read or write operations performed during the $\lit{insert}(v)$ or $\lit{remove}(v)$ returned $\bot$,
appropriate matching responses are returned as prescribed by the sequential implementation $\LL$.
Any update operation of $I^{RM}$ uses at most two expensive synchronization
patterns~\cite{AGK11-popl}.

Theorem~\ref{th:lr} in Appendix~\ref{app:lltm} shows that $I^{RM}$ is $(\LL,\ms{set})$-LSL.
The pseudocode in Algorithm~\ref{alg:elastic} is given for managed languages as there is no explicit garbage collector, 
but one could add 
an epoch-based garbage collector that deallocates a node as soon as all operations concurrent with its removal are complete.

Now we show that $I^{RM}$ supersedes, in terms of concurrency, \emph{any} implementation in classes
$\mathcal{P}$ or $\mathcal{SM}$.
The proof is based on a more general optimality result, interesting in its own right: 
any finite schedule rejected by $I^{RM}$ is not \emph{observably
LS-linearizable} (or simply \emph{observable}). 
A schedule $\sigma$ is observable 
if it has an extension $\sigma'$ such that for all $v\in \mathbb{Z}$, 
$\sigma'$ extended with a complete execution of $\lit{contains}(v)$ that returns a matching response is LS-linearizable. 
Intuitively, a schedule is observable if it incurs no lost updates.
One example of a non-observable schedule is $\sigma'$ in
Figure~\ref{fig:ex2}(b): since one of the two concurrent
updates overwrites the effect of the other, $\sigma'$ extended with a
complete execution of $\lit{contains}(1)$ 
is not linearizable with respect to $\ms{set}$.
\begin{theorem}[Optimality]
\label{th:lrelaxed}
$I^{RM}$ accepts all schedules that are observable with respect to $(\ms{LL},\ms{set})$.
\end{theorem}
\begin{proofsketch}
We prove that any schedule rejected by $I^{RM}$ is  
not observable.
We go through the cases when a read or write returns $\bot$ (implying the operation fails to return a matching response) and
thus the current schedule is rejected:
(1) $\lit{read}(X_{\ell})$ returns $\bot$ in line~\ref{line:elastic:abort1}
when $r[\ell]=\true$ or when $\ms{ver}_1\neq \ms{ver}_2$, (2)
$\lit{write}(X_{\ell})$ performed by $\lit{remove}(v)$ either returns $\bot$
in line~\ref{line:linw} when the $\lit{cas}$ operation on
$L[\ell]$ returns $\false$ or returns $\bot$
in line~\ref{line:linw2} when the $\lit{cas}$ operation on the element that
stores $v$ returns $\false$, and (3) $\lit{write}(X_{\ell})$ performed by $\lit{insert}$ returns $\bot$
in line~\ref{line:inswrite} when the $\lit{cas}$ operation on
$L[\ell]$ returns $\false$.

Consider the subcase (1a), $r[\ell]$ is set $\true$ by a
preceding or concurrent $\lit{write}(X_{\ell})$ (line~\ref{line:elastic:gc}).
The high-level operation performing this $\lit{write}$ is a $\lit{remove}$ 
that marks the corresponding list element as removed. 
Since no removed element can be read in a sequential execution
of $LL$, the corresponding history is not locally serializable.
Alternatively, in subcase (1b), the version of $X_{\ell}$ read previously in line~\ref{line:rver}
has changed. Thus, an update operation has concurrently performed a write to $X_{\ell}$.
However, there exist executions that export such schedules.

In case $(2)$, the $\lit{write}$ performed by a $\lit{remove}$ operation returns $\bot$.
In subcase (2a), $X_{\ell}$ is currently
locked. 
Thus, a concurrent high-level operation has previously locked $X_{\ell}$ (by successfully
performing $L[\ell].\lit{cas}()$ in line~\ref{line:linw}) and has not yet released
the lock (by writing $\tup{\ms{ver}', \lit{false}}$ to
$L[\ell]$ in line~\ref{line:release}). 
In subcase (2b), the current version of $X_{\ell}$ (stored in $L[\ell]$) differs 
from the version of $X_{\ell}$ witnessed by a preceding $\lit{read}$. 
Thus, a concurrent high-level operation completed a write to $X_{\ell}$
\emph{after} the current high-level operation $\pi$ performed a $\lit{read}$ of $X_{\ell}$.
In both (2a) and (2b), a concurrent high-level updating operation $\pi'$
($\lit{remove}$ or $\lit{insert}$) has written or is about to perform a $\lit{write}$ to $X_{\ell}$.  
In subcase (2c), the $\lit{cas}$ on the element $X_{\ell'}$ (element that stores the value $v$) executed by $\lit{remove}(v)$ 
returns $\false$ (line~\ref{line:linw2}).
Recall that by the sequential implementation $\LL$, 
$\lit{remove}(v)$ performs a $\lit{read}$ of
$X_{\ell'}$ prior to the $\lit{write}(X_{\ell})$, where $X_{\ell}.\textit{next}$
refers to $X_{\ell'}$.
If the \emph{cas} on $X_{\ell'}$ fails, there exists a process 
that concurrently performed a $\lit{write}$ to $X_{\ell'}$, but
after the $\lit{read}$ of $X_{\ell'}$ by $\lit{remove}(v)$.
In all cases, we observe that if we did not abort the write to $X_{\ell}$, then
the schedule extended by a complete execution of $\lit{contains}$ is not LSL.

In case $(3)$, the $\lit{write}$ performed by an $\lit{insert}$ operation returns $\bot$. Similar arguments to
case $(2)$ prove that any schedule rejected is not observable LSL.
\end{proofsketch}
Theorem~\ref{th:lrelaxed} implies that the schedules exported by the
histories in Figures~\ref{fig:ex1} and \ref{fig:ex2}(a) and that are not
accepted by any $I'\in \mathcal{SM}$ and any $I\in \mathcal{P}$, respectively,
are indeed accepted by $I^{RM}$.
But it is easy to see that implementations in $\mathcal{SM}$ and $\mathcal{P}$ can only
accept observable schedules.  
As a result, $I^{RM}$ can be shown to strictly supersede any
pessimistic or TM-based implementation of the list-based set.  
\begin{corollary}
\label{cr:mrp}
$I^{RM}$ accepts every schedule accepted by any implementation in $\mathcal{P}$ and $\mathcal{SM}$.
Moreover, $I^{RM}$ accepts schedules $\sigma$ and $\sigma'$ that are rejected by any
implementation in $\mathcal{P}$ and $\mathcal{SM}$, respectively. 
\end{corollary}
One take-away from these results is that generic optimistic implementations,
appropriately relaxed, are able to provide strictly more concurrency
than pessimistic or strongly consistent optimistic ones.
Our implementation $I^{RM}$ is in fact optimal with respect to
concurrency, while still incurring minimal cost in terms of step-complexity and
use of expensive synchronization patterns. 
%
%
%
%
\input{related}
\input{conc}
\bibliographystyle{abbrv}
\newpage
\bibliography{references}

\appendix
\input{appendix}

\end{document}

%% file: intro.tex
\section{Introduction}
\label{sec:intro}
To exploit concurrency provided by modern multi-cores, conventional lock-based synchronization 
pessimistically protects accesses to  the shared memory before
executing them.
Speculative synchronization, achieved using transactional memory (TM) or
conditional primitives, such as CAS or LL/SC,  
optimistically executes memory operations with a risk of aborting them in the future.
A programmer typically uses these synchronization techniques 
as ``wrappers'' to allow every process (or thread) to \emph{locally} run its sequential code while ensuring 
that the resulting concurrent execution is \emph{globally} correct.

Unfortunately, it is difficult for programmers to tell in advance 
which of the 
techniques
will establish more concurrency in their resulting programs.
By speculatively executing concurrent accesses that would have to
block in a lock-based implementation,
TMs~\cite{HM93,ST95,HLR10} seemingly provide high
concurrency.   
However, TMs conventionally ensure 
\emph{serializability}~\cite{Pap79-serial} or even
stronger properties~\cite{tm-book}, which may prohibit concurrent
scenarios allowed by the sequential specification of the specific data
structure we intend to implement~\cite{GG14}.

In this paper, we analyze the ``amount of concurrency'' one can obtain by turning a sequential program into a concurrent one.
In particular, we compare the use of 
optimistic and pessimistic
synchronization techniques,  whose popular examples are transactions and locking, respectively.
To fairly compare concurrency provided by implementations
based on various techniques,    
one has (1)~to  define what it means for a concurrent program to be
correct regardless of the type of  synchronization it uses and 
(2)~to define a metric of concurrency. 

\vspace{1mm}\noindent\textbf{Correctness.}
We begin by defining a novel consistency criterion,
namely \emph{locally-serializable linearizability}.
We say that a concurrent implementation of a given sequential data type is
\emph{locally serializable} if it
ensures that the local execution of each 
operation 
is equivalent to \emph{some} execution of its sequential implementation.
This condition is weaker than serializability
since it does not require that there exists a \emph{single} sequential 
execution  that is consistent with all local executions.
It is however sufficient to guarantee that optimistic
executions do not observe an inconsistent transient state that could 
lead, for example, to a fatal error like division-by-zero.

Furthermore, the implementation should ``make sense'' globally, 
given the \emph{sequential type} of the data structure we implement.
The high-level history of every execution 
of a concurrent implementation must be 
\emph{linearizable}~\cite{HW90,AW04} with respect to 
this sequential type.
The combination of local serializability and linearizability gives
a correctness criterion that we call \emph{\LS-linearizability},
where {\LS} stands for ``locally serializable''.
We show that LS-linearizability is, as the original  linearizability,
compositional~\cite{HW90,HS08-book}: a composition of LS-linearizable 
implementations is also LS-linearizable. 
Unlike linearizability, however, it is not non-nonblocking: local
serializability may prevent an operation in a finite LS-linearizable
history from completing in a non-blocking manner. 

We apply the criterion of LS-linearizability to  
two broad classes of \emph{pessimistic} and \emph{optimistic}
synchronization techniques. 
Pessimistic implementations capture what can be achieved 
using classic conservative locks like mutexes, 
spinlocks, reader-writer locks.
In contrast, optimistic implementations proceed speculatively and
may roll back in the case of conflicts, e.g.,  
relying on classical 
TMs, like TinySTM~\cite{FFR08} or NOrec~\cite{norec},  or more relaxed forms of
optimistic techniques, such as ``lazy''
synchronization~\cite{HHL+05},  elastic transactions~\cite{FGG09} or
view transactions~\cite{AMT10}.

\vspace{1mm}\noindent\textbf{Concurrency metric.}
We measure the amount of concurrency provided by an LS-linearizable implementation as the set of schedules it accepts.
To this end, we define a concurrency metric 
inspired by the analysis of parallelism in database concurrency control~\cite{Yan84,Her90}
and transactional memory~\cite{GHF10}.
More specifically, we assume an external scheduler that defines which
processes execute which steps of the corresponding sequential program 
in a dynamic and unpredictable fashion. 
This allows us to define concurrency provided by an implementation as the set of \emph{schedules} 
(interleavings of steps of concurrent sequential operations) 
it \emph{accepts} (is able to effectively process).
%

Our concurrency metric is platform-independent and it allows for
measuring relative concurrency of LS-linearizable implementations
using arbitrary synchronization techniques.   
We do not claim that this metric necessarily captures 
efficiency, as it does not account for other factors, 
like cache sizes, cache coherence protocols, or computational costs of 
validating a schedule, which may also affect performance on
multi-core architectures.
However, our experimental evaluations show that the gain in concurrency 
may translate into better scalability.

\vspace{1mm}\noindent\textbf{Measuring concurrency.}
This paper provides 
a framework to compare the concurrency one can get
by choosing a particular synchronization technique for a specific data type.
For the first time, we analytically capture the
inherent incomparability of TM-based and pessimism-based
implementations in exploiting concurrency.
We illustrate this using a popular sequential list-based set
implementation~\cite{HS08-book}, concurrent implementations of which
are our running examples.
More precisely, we show that there exist TM-based implementations that, for some workloads, 
allow for more concurrency than \emph{any} pessimistic implementation,
but we also show that there exist pessimistic implementations that, for other workloads, allow for more 
concurrency than \emph{any} TM-based implementation.

Intuitively, an implementation based on transactions 
may abort an operation based on the way
concurrent steps are scheduled, 
while a pessimistic implementation 
has to proceed eagerly without knowing about how future steps will be 
scheduled, sometimes over-conservatively rejecting 
a potentially acceptable schedule.  
By contrast, pessimistic implementations designed to exploit
the semantics of the data type can supersede the
``semantics-oblivious'' TM-based implementations.

More surprisingly, we demonstrate that combining the benefit of pessimistic implementations, 
namely their semantics awareness, 
and the benefit 
of transactions,
namely their optimism, 
enables implementations that are strictly 
better-suited for exploiting concurrency 
than any of them individually.
We describe a generic optimistic implementation of 
the list-based set that is \emph{optimal} with respect to our
concurrency metric: we show that, essentially, it accepts \emph{all} 
correct concurrent schedules.
Our implementation, designed with our theoretical concurrency metric
in mind,  is surprisingly reminiscent of the
state-of-the-art ``pragmatic'' list-based set implementations~\cite{HHL+05,harris-set}. 
Indeed, our experimental results confirm
that optimal concurrency leads to higher performance
than popular pessimistic algorithms, like hand-over-hand list-based
sets, or generic TM-based optimistic ones.

Our concurrency analysis is focused on a specific example of a
list-based set, but our findings demonstrate the potential
of the concurrency-based  approach in analyzing and comparing wider classes of
LS-linearizable data structures. 

%% file: simple-elastic-s3.tex
\begin{algorithm*}[t]
\caption{Code for process $p_k$ implementing reads and writes in implementation $I^{RM}$}
\label{alg:elastic}
  \begin{algorithmic}[1]
  	\begin{multicols}{2}
	{\size
	
	\Part{Shared variables}{
	  \State for each object $X_{\ell}$:
	  \State ~~~~~~$\ms{t-var}[\ell]$, initially 0
	  \State ~~~~~~$r[\ell]$, initially $\false$
	  \State ~~~~~~$L[\ell]\in \N \times \{\lit{true}, \lit{false}\}$ supports $\lit{read}$
	  \State ~~~~~~~~~~~$\lit{write}, \lit{cas}$ operations, initially $\tup{0, \lit{false}}$
	}\EndPart	
	
   	\Statex
	  
	\Part{Local variables of process $p_k$}{
	  \State $\ms{rbuf}_k[i] \subset X \times \N$; $i=\{1,2\}$ cyclic buffer of size $2$,  
	  \State ~~~initially $\emptyset$
	}\EndPart	
	
	\Statex
%
	
   	\Part{$\lit{read}_k(X_{\ell})$ executed by $\lit{insert}$, $\lit{remove}$, $\lit{contains}$}{	  
	  \State $\tup{\ms{ver_1}, *} \gets L[\ell].read()$ \label{line:rver} \Comment{get versioned lock}
	  \State $val \gets \ms{t-var}[\ell].read()$ \Comment{get value} \label{line:linr}
	  \State $r \gets r[\ell].read()$ \label{line:reread}
	  \State $\tup{\ms{ver_2}, *} \gets L[\ell].read()$ \Comment{reget versioned lock}
	  \If{$(\ms{ver}_1 \neq \ms{ver}_2) \vee r$} \label{line:flagcheck}
	  	\Return $\bot$ \label{line:elastic:abort1} \EndReturn
	  \EndIf
	  \State $\ms{rbuf}_k.\lit{add}(\tup{X_{\ell},\ms{ver}_{1}})$
          \Comment{override penultimate entry}
	  \Return $\ms{val}$ \EndReturn \label{line:tx-read}
	}\EndPart
	  
	
	\newpage

	\Part{$\lit{write}_k(X_{\ell},v)$ executed by $\lit{remove}$}{
	  \State {\bf let} $\ms{oldver}_\ell$ be such that $\tup{X_\ell, \ms{oldver}_\ell} \in \ms{rbuf}_k$
	  \State $\ms{ver} \gets \ms{oldver}_\ell$
	  \If{$ \neg L[\ell].cas(\tup{\ms{ver}, \lit{false}}, \tup{\ms{ver}, \lit{true}})$} \label{line:linw}
	  		\Return $\bot$  \label{line:elastic:abort2}\Comment{grab lock or abort}	  	\EndReturn
	  \EndIf
	  \State {\bf let} $X_{\ell'}\neq X_{\ell}$ be such that $\{ X_{{\ell}^\prime},ver_{\ell'} \}\in \ms{rbuf}_k$
	  \If{$ \neg L[\ell'].cas(\tup{\ms{ver}_{\ell'}, \lit{false}}, \tup{\ms{ver}_{\ell'}, \lit{true}})$} \label{line:linw2}
	  		\Return $\bot$  \label{line:elastic:abort2}\Comment{grab lock or abort}	  	\EndReturn
	  \EndIf
	  \State $r[\ell^\prime].write(\true)$ \label{line:elastic:gc}  \Comment{mark element for deletion}
	      
	  \State $\ms{t-var}[\ell].write(v)$ \Comment{update memory} \label{line:commit-update}
	  \State $L[\ell].\lit{write}(\tup{\ms{ver}+1, \lit{false}})$\Comment{release locks} \label{line:release}
	  \State $L[\ell'].\lit{write}(\tup{\ms{ver}_{\ell'}+1, \lit{false}})$ \label{line:release2}
	  \Return $\lit{ok}$ \EndReturn
	}\EndPart

	\Statex
	\Part{$\lit{write}_k(X_{\ell},v)$ executed by $\lit{insert}$}{
	  \State {\bf let} $\ms{oldver}_\ell$ be such that $\tup{X_\ell, \ms{oldver}_\ell} \in \ms{rbuf}_k$
	  \State $\ms{ver} \gets \ms{oldver}_\ell$
	  \If{$ \neg L[\ell].cas(\tup{\ms{ver}, \lit{false}}, \tup{\ms{ver}, \lit{true}})$} \label{line:inswrite}
	  		\Return $\bot$  \label{line:elastic:abort3}\Comment{grab lock or abort}	  	\EndReturn
	  \EndIf
      \State $\ms{t-var}[\ell].write(v)$ \Comment{update memory} \label{line:inscommit}
      \State $L[\ell].\lit{write}(\tup{\ms{ver}+1, \lit{false}})$\Comment{release locks} \label{line:insrel}
	  \Return $\lit{ok}$ \EndReturn
	}\EndPart

  }
  \end{multicols}
  \end{algorithmic}\vspace{-1em}
\end{algorithm*}

%

%% file: related.tex
\section{Related work}
\label{sec:related}
%
%
Sets of accepted schedules are commonly used as a
metric of concurrency provided by a shared-memory
implementation.
For static database transactions, 
Kung and Papadimitriou~\cite{KP79} use the metric to 
capture the parallelism of a locking scheme,
While acknowledging that the metric is theoretical, they 
insist that it may
have ``practical significance as
well, if the schedulers in question have relatively small
scheduling times as compared with waiting and execution
times.'' 
Herlihy~\cite{Her90} employed the metric to compare various
optimistic and pessimistic synchronization techniques using
commutativity
of operations constituting high-level transactions.   
A synchronization technique is implicitly considered in~\cite{Her90} as highly
concurrent, namely ``optimal'',
if no other technique accepts more schedules. 

By contrast, we focus here on a \emph{dynamic} model where the scheduler cannot 
use the prior knowledge of all the shared addresses to be accessed. 
Also, unlike~\cite{KP79,Her90}, 
we require \emph{all} operations, including aborted ones, to observe (locally) consistent states.
As we confirm experimentally, 
our (provably optimal) optimistic implementation incurs negligible scheduling overhead, which makes the motivation of the metric proposed
in~\cite{KP79} applicable. 

Gramoli et al.~\cite{GHF10} defined a concurrency metric, the \emph{input
acceptance}, as the ratio of committed transactions over aborted
transactions when TM executes the given schedule.   
Unlike our metric, input acceptance does not apply to
lock-based programs. 

\ignore{ 
A quite different but popular concurrency notion, called
\emph{disjoint-access parallelism}~\cite{israeli-disjoint}, 
was recently considered in the TM context~\cite{AHM09,EFKMT12}.
A TM is disjoint-access parallel (DAP) if it guarantees 
that two transactions accessing the same meta-data also access 
the same transactional object. 
Our concurrency metric not only applies to lock-based implementations but is also
more fine-grained ones \petr{did not get this point}
as it allows to determine relative concurrency provided 
by different DAP implementations.
}

Similar to other relaxations of opacity~\cite{tm-book} like \emph{TMS1}~\cite{TMS09} and \emph{VWC}~\cite{damien-vw},  
safe-strict serializable implementations ($\mathcal{SM}$) require that every transaction (even aborted and incomplete) observes
``correct'' serial behavior. 
However, unlike TMS1, we do not require the \emph{local} serial executions
to always respect the real-time order among transactions. 
Unlike VWC, we model transactional
operations as intervals with an invocation and a response and does not
assume unique writes (needed to define causal past in VWC). 
Therefore, $\mathcal{SM}$ appears weaker than both TMS1 and VWC.
Though weak and possibly not very pragmatic, it still allows us 
to show that resulting TM-based LSL implementations 
reject some schedules accepted by pessimistic locks. 
However, we can easily extend our results to show that even opaque TMs accept
some schedules rejected by any pessimistic algorithm.   
    
The problem of transforming a
sequential implementation of a list-based set into a concurrent
one was considered before in special settings.
Vechev and Yahav~\cite{vechev08} considered using locks while Felber et al.~\cite{FGG09} considered using elastic transactions.
Our framework applies to generic
concurrent transformations of sequential implementations
using arbitrary synchronization techniques. 
%
%
%
%
%

%% file: conc.tex
\section{Discussion and concluding remarks}
\label{sec:conc}
To confirm the practicality of our optimistic list-based set, $I^{RM}$,
we compared its Java implementation against a pessimistic
implementation \emph{HOHL}
that uses hand-over-hand locking (we adopted the Java pseudocode by Herlihy and 
Shavit~\cite[Chapter 9]{HS08-book}). 
Figure~\ref{fig:alg2-vs-hoh-10} presents the throughput 
(the number of completed operations per millisecond) 
of the two algorithms on a $32$-way machine 
where up to $32$ threads run $5$\% updates 
(either $\lit{remove}$ or $\lit{insert}$ with the same probability) 
and $95$\% $\lit{contains}$ operations on a list, initially populated
with $512$ integer values. 
$I^{RM}$ 
outperforms both HOHL and E-STM~\cite{FGG09}, and 
the reason for this could be 
that $I^{RM}$ is optimal in terms of concurrency (Theorem~\ref{th:lrelaxed}), 
while the HOHL is serializable~\cite{ARR10} and, thus, rejects
large classes of correct schedules (e.g., of the kind of 
$\sigma_0$ in~Figure~\ref{fig:ex1}). 
E-STM does not provide optimal concurrency but rejects less schedules than HOHL.
Additionally, our implementation uses asymptotically less expensive
memory barriers and read-modify-write primitives~\cite{AGK11-popl}
than HOHL.
We deduce that accepting all observable schedules may be 
quite efficient on some applications.
For the sake of simplicity, we did not optimize our code, e.g., 
by removing wrappers or using partial aborts.
(More experimental results are given in Appendix~\ref{sec:expe}.)

\begin{floatingfigure}[right]{0.27\paperwidth}
{\small
\hspace{-2em}\includegraphics[clip=true,viewport=5 0 375 112,scale=.9]{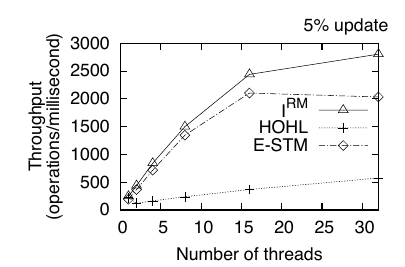}
   \caption{$I^{RM}$ vs. HOHL\label{fig:alg2-vs-hoh-10}}
}\end{floatingfigure}
%

$I^{RM}$ is surprisingly reminiscent of the
list-based set implementations in~\cite{HHL+05} and \cite{harris-set}
(state-of-the-art, to the best of our knowledge).
However, because of specific optimizations of the $\lit{contains}$
operation, strictly speaking, none of these two algorithms is locally
serializable, and thus LSL.
The implementations in~\cite{HHL+05,harris-set}
use the logical deletion technique to associate a \emph{marked} field to an element to indicate if it
is contained in the list.
The $\lit{contains}$ of the lazy list-based set~\cite{HHL+05} may read an element marked for deletion, whereas
the $\lit{contains}$ of Harris list-based set~\cite{harris-set} even uses $\lit{cas}$ to remove logically deleted nodes.
But the apparent similarities between our $I^{RM}$ and the algorithms
in~\cite{HHL+05,harris-set} suggest
that looking for a concurrency optimal LS-linearizable
implementation also helps in optimizing performance.  
An implementation, once proven to be concurrency optimal, may
be optimized further to boost performance, as is possible with $I^{RM}$.

We derived our concurrency lower bounds in the context of the list-based set, a data structure
that is suitable for exploiting concurrency because of its localized
updates, but we believe that it should be possible to generalize our results to a wider class of search structures. 
This paper provides some preliminary hints in the quest for 
the ``right'' synchronization technique to develop highly concurrent and efficient implementations of data 
types.
Our results are relevant to programmers leveraging
multi-core architectures as well as to computer manufacturers 
who aim at defining new instruction sets.
This work suggests directions to identify the ``killer'' application
for existing and emerging synchronization techniques.

%% file: appendix.tex
\section{Sequential implementation of the set type}
\label{app:seq}
Recall that an \emph{object type} $\tau$ is a tuple
$(\Phi,\Gamma, Q, q_0, \delta)$ where
$\Phi$ is a set of operations,
$\Gamma$ is a set of responses, $Q$ is a set of states, $q_0\in Q$ is an
initial state and 
$\delta \subseteq Q\times \Phi \times Q\times \Gamma$ 
is a transition relation that determines, for each state,
and each operation, the set of possible
resulting states and  produced responses. 
Hence, $(q,\pi,q',r) \in \delta$ implies that when
an operation $\pi \in \Phi$ is applied on an object of type $\tau$
in state $q$, the object moves to state $q'$ and returns a response $r$.
We consider only types that are \emph{total}
 i.e., for every $q\in Q$,
$\pi \in \Phi$, there exist  $q' \in Q$ and
$r\in \Gamma$ such that $(q,\pi,q',r) \in \delta$.
We assume that every type $\tau=(\Phi,\Gamma, Q, q_0, \delta)$ is \emph{computable}, i.e., 
there exists a Turing machine that, 
for each input $(q,\pi)$, $q \in Q$, $\pi\in \Phi$, computes
a pair $(q',r)$ such that $(q,\pi,q',r) \in \delta$.

Formally, the \emph{set} type is defined by the tuple $(\Phi,\Gamma, Q, q_0, \delta)$ where:
\begin{enumerate}
\item[$\Phi$] $=\{\lit{insert}(v), \lit{remove}(v), \lit{contains}(v) \}$; $v \in \mathbb{Z}$ 
\item[$\Gamma$] $=\{\true,\false\}$ 
\item[$Q$] is the set of all finite subsets of $\mathbb{Z}$; $q_0=\emptyset$
\item[$\delta$] is defined as follows:
\begin{enumerate}
\item[$(1)$:]
$(q,\lit{contains}(v),q,(v\in q))$
\item[$(2)$:]
$(q,\lit{insert}(v),q \cup \{v\},(v \not\in q))$
\item[$(3)$:]
$(q,\lit{remove}(v),q \setminus \{v\},(v \in q))$
\end{enumerate}
\end{enumerate}
%
%
%
\input sequential-fix.tex
The sequential implementation \LL~ of the \emph{set} type is
presented in Algorithm~\ref{alg:lists}. The implementation uses a \emph{sorted linked list} data structure 
in which each element (except the \emph{tail}) maintains a \textit{next} field to provide a pointer to the
successor node. Initially, the \emph{next} field of the \emph{head} element points to \emph{tail}; \emph{head}
(resp. \emph{tail}) is initialized with values $-\infty$ (resp. $+\infty$) that is smaller (resp. greater) than
the value of any other element in the list.
%

\section{\LS-linearizability is compositional}
\label{app:comp}
We define the composition of two distinct object types $\tau_1$ and $\tau_2$ 
as a type $\tau_1\times\tau_2=(\Phi,\Gamma,Q,q_0,\delta)$ as follows: 
$\Phi=\Phi_1\cup \Phi_2$, $\Gamma=\Gamma_1\cup 
\Gamma_2$,\footnote{Here we treat each $\tau_i$ as a distinct type by adding
index $i$ to all elements of $\Phi_i$, $\Gamma_i$, and $Q_i$.}   
$Q=Q_1\times Q_2$,
$q_0=({q_0}_1,{q_0}_2)$, and  $\delta \subseteq Q\times \Phi \times Q\times
\Gamma$ is such that $((q_1,q_2),\pi,(q_1'q_2'),r)\in\delta$ if and only if for $i\in \{1,2\}$, if 
$\pi\in \Phi_i$ then $(q_i,\pi,q_i',r)\in\delta_i$ $\wedge$ $q_{3-i}=q^{\prime}_{3-i}$.

Every sequential implementation $\id{IS}$ of an object  $O_1\times O_2$ of a
composed type $\tau_1\times\tau_2$ naturally induces two sequential
implementations $I_{S1}$ and $I_{S2}$ of objects $O_1$ and $O_2$,
respectively. 
Now a correctness criterion 
$\Psi$
is \emph{compositional} if for every
history $H$ on an object composition $O_1\times O_2$, 
if 
$\Psi$
holds for $H|O_i$ with
respect to $I_{Si}$, for $i \in \{1,2\}$, then
$\Psi$
holds for $H$ with
respect to $\id{IS}=I_{S1}\times I_{S2}$.
Here, $H|O_i$ denotes the subsequence of $H$ consisting of events on $O_i$.
\begin{theorem}
\label{th:comp}
\LS-linearizability is compositional. 
\end{theorem}
\begin{proof}
Let $H$, a history on $O_1\times O_2$,  be \LS-linearizable
with respect to $\id{IS}$. 
Let each $H|O_i$, 
$i\in\{1,2\}$, 
be \LS-linearizable with respect to $I_{Si}$. 
Without loss of generality,  we assume that $H$ is complete (if $H$
is incomplete, we consider any completion of it containing
\LS-linearizable completions of  $H|O_1$ and $H|O_1$).

Let $\tilde H$ be a completion of the high-level history corresponding to $H$ such that
$\tilde H|O_1$ and $\tilde H|O_2$ are linearizable with respect to $\tau_1$
and $\tau_2$, respectively. Since linearizability is
compositional~\cite{HW90,HS08-book}, $\tilde H$ is linearizable with respect to $\tau_1\times\tau_2$.

Now let, for each operation $\pi$, $S_{\pi}^1$ and $S_{\pi}^2$ be any two sequential histories of
$I_{S1}$ and $I_{S2}$  such that $H|\pi|O_j=S_{\pi}^j|\pi$, for $j \in \{1,2\}$
(since 
$H|O_1$
and $H|O_2$ are \LS-linearizable such histories exist).
We construct a sequential history $S_{\pi}$ by interleaving events of
$S_{\pi}^1$ and $S_{\pi}^2$ so that $S_{\pi}|O_j=S_{\pi}^j$, 
$j\in\{1,2\}$.
Since each $S_{\pi}^j$ acts on a distinct component $O_j$ of $O_1\times
O_2$, every such $S_{\pi}$ is a sequential history of $\id{IS}$.
We pick one $S_{\pi}$ that respects the local history $H|\pi$,
which is possible, since $H|\pi$ is consistent with both    
$S_1|\pi$ and $S_2|\pi$. 

Thus, for each $\pi$, we obtain a history of $\id{IS}$ that agrees with
$H|\pi$. Moreover, the high-level history of $H$ is linearizable. Thus, $H$ is \LS-linearizable with respect to $\id{IS}$. 
\end{proof}


\section{Formal proofs}
\label{app:proofs}
\subsection{$\mathcal{SM}$ vs. $\mathcal{P}$}
\label{app:smp}
\paragraph{A pessimistic implementation $I^H\in \mathcal{P}$ of $(\LL,\ms{set})$.}
The implementation $I^H$ is a lock-based implementation that associates every object with a 
distinct \emph{lock} and another base object that stores the \emph{value} of the object.
In $I^H$, the $\lit{contains}$ operation uses shared \emph{hand-over-hand 
locking}~\cite{BS88,HS08-book}.
Each read operation performed by $\lit{contains}$ 
acquires the \emph{shared lock} on the object, 
reads the $\ms{next}$ field of the element before releasing the shared lock on the predecessor element. 
Update operations ($\lit{insert}$ and
$\lit{remove}$) acquire the \emph{exclusive lock} on the $\ms{head}$ during $\lit{read}(\ms{head})$ 
and release it at the end. Every other read operation performed 
by an update simply reads the $\ms{next}$ field of the element to traverse the list. The write operation
performed by a $\lit{insert}$ or $\lit{remove}$ acquires the exclusive lock, writes the value
to the element and releases the lock.
There is no real concurrency between any two update operations since the process holds the 
exclusive lock on the $\ms{head}$ throughout the operation execution.
Intuitively, it is easy to observe that $I^H$ is LS-linearizable with respect to $(\LL,\ms{set})$.

\begin{theorem}[Part 1 of  Theorem~\ref{th:mpl}]
There exists a schedule $\sigma_0$ of $(\LL,\ms{set})$ that is
accepted by $I^H\in \mathcal{PL}$, but not accepted by \emph{any}
$(\LL ,\ms{set})$-LSL implementation $I\in \mathcal{SM}$.
\end{theorem}
\begin{proof}
Let $\sigma_0$ be the schedule of $(\LL ,\ms{set})$ depicted in Figure~\ref{fig:ex1}.
Suppose by contradiction that $\sigma_0 \in\S(I)$, where $I$ is an implementation of $(\LL, \ms{set})$ based on any safe-strict serializable TM.
Thus, there exists an execution $\alpha$ of $I$ that exports $\sigma_0$.
Now consider two cases:
\begin{itemize}
\item
Suppose that the read of $X_4$ by $\lit{contains}(5)$ 
returns the value of $X_4$ that is updated by $\lit{insert}(5)$.
Since $\lit{insert}(2) \rightarrow_{\alpha} \lit{insert}(5)$, 
$\lit{insert}(2)$ must precede $\lit{insert}(5)$ in any sequential execution $\alpha'$ equivalent to $\alpha$. 
Also, since $\lit{contains}(5)$ reads $X_1$ prior to its update by $\lit{insert}(2)$, 
$\lit{contains}(5)$ must precede $\lit{insert}(2)$ in $\alpha'$. 
But then the read of $X_4$ is not legal in $\alpha'$---a contradiction since $\alpha$ must be serializable.
\item
Suppose that $\lit{contains}(5)$
reads the initial value of $X_4$, i.e., its value prior to the write to $X_4$ by $\lit{insert}(5)$, 
where $X_4.\ms{next}$ points to the
\emph{tail} of the list (according to our sequential implementation $\LL$).
But then, according to $\LL$, $\lit{contains}(5)$ cannot access
$X_5$ in $\sigma_0$---a contradiction.  
\end{itemize}
Consider the pessimistic implementation $I^H \in \mathcal{P}$:
since the $\lit{contains}$ operation traverses the list using shared hand-over-hand locking, 
the process $p_i$ executing $\lit{contains}(5)$
can release the lock on element $X_1$ prior to the acquisition of the exclusive lock on $X_1$ by $\lit{insert}(2)$.
Similarly, $p_i$ can acquire the shared lock on $X_4$ immediately after the release of the 
exclusive lock on $X_4$ by the process executing $\lit{insert}(5)$ while still holding 
the shared lock on element $X_3$. Thus, there exists an execution of $I^H$ that exports $\sigma_0$.
\end{proof}%
\paragraph{An optimistic implementation $I^C\in \mathcal{SM}$ of $(\LL,\ms{set})$.}
Recall that $\mathcal{SM}$ denotes the set of concurrent implementations based on TMs
that ensure the following safety condition in every execution.

Let $\alpha$ denote the execution of a TM-based implementation and
$\ms{ops}(\alpha)$, 
the set of operations each of which performs at least one event in $\alpha$.
Let ${\alpha}^k$ denote the prefix of $\alpha$ up to the last event of operation $\pi_k$.
Let $\ms{cseq}(\alpha)$ denote the subsequence of ${\alpha}$  that
consists of the events of the complete operations in $\alpha$. 
We say that $\alpha$ is \emph{strictly serializable} if 
there exists a legal sequential execution $\alpha'$ equivalent to
$\ms{cseq}(\alpha)$
such that $\rightarrow_{\ms{cseq}(\alpha)} \subseteq \rightarrow_{\alpha'}$. 

An execution $\alpha$ of a TM-based implementation is 
\emph{safe-strict serializable} if
(1) $\alpha$ is strictly serializable, and
(2) for each operation $\pi_k$ that is incomplete or returned $\bot$ in
$\alpha$, there exists a legal sequential execution of operations
$\alpha'=\pi_0\cdots \pi_i\cdot \pi_k$ such that $\{\pi_0,\cdots, \pi_i\}
\subseteq \ms{ops}(\ms{cseq}(\alpha^k))$ and $\forall \pi_m\in \ms{ops}(\alpha'):{\alpha'}|m={\alpha^k}|m$.

The implementation $I^C$ is
based on a TM that ensures the following condition in every execution~\cite{GK09-progressiveness}: 
if a transaction $T_i$ aborts, then it encounters a conflict with a transaction $T_j$, i.e.,  
$T_i$ and $T_j$ are concurrent, both access the same object $X$, and
at least one of these is a $\lit{write}$.
An implementation of such a safe-strict serializable TM can be
found in~\cite{KR11}. 
\begin{theorem}[Part 2 of Theorem~\ref{th:mpl}]
There exists a schedule $\sigma$ of $(\LL,\ms{set})$ 
that is accepted by an $(\LL ,\ms{set})$-LSL  implementation $I^C \in \mathcal{SM}$,
but not accepted by \emph{any} $(\LL ,\ms{set})$-LSL implementation in $\mathcal{P}$.
\end{theorem}
\begin{proof}
We show first that the schedule $\sigma$ of $(\LL,\ms{set})$ depicted
in Figure~\ref{fig:ex2}(a) is not accepted by any implementation in $\mathcal{P}$.
Suppose the contrary and let $\sigma$ be exported by an execution $\alpha$. 
Here $\alpha$ starts with three sequential $\lit{insert}$ operations with
parameters $1$, $2$, and $3$. The resulting ``state'' of the set is
$\{1,2,3\}$, where value $i\in \{1,2,3\}$ is stored in object $X_i$.   

Suppose, by contradiction, that some $I\in \mathcal{P}$ accepts $\sigma$. 
We show that $I$ then accepts the schedule $\sigma'$ depicted in Figure~\ref{fig:ex2}(b), which starts with a sequential
execution of $\lit{insert}(3)$ storing value $3$ in object $X_1$. 

Let $\alpha'$ be any history of $I$ that exports $\sigma'$.
Recall that we only consider obedient implementations:
in $\alpha'$: the read of \emph{head} by $\lit{insert}(2)$ in $\sigma'$ refers to $X_1$ (the next element to be read by $\lit{insert}(2)$). 
In $\alpha$, element $X_1$ stores value $1$,
i.e., $\lit{insert}(1)$ can safely return $\lit{false}$, while in
$\sigma'$, $X_1$ stores value $3$, i.e., the next step of
$\lit{insert}(1)$ must be a write to \emph{head}.       
Thus, no process
can distinguish $\alpha$ and $\alpha'$ 
before the
read operations on $X_1$ return. 
Let $\alpha''$ be the prefix of $\alpha'$ ending with $R(X_1)$
executed by $\lit{insert}(2)$.
Since $I$ is deadlock-free, we have an extension of $\alpha''$ in
which both $\lit{insert}(1)$ and $\lit{insert}(2)$ terminate; we show that this extension violates linearizability. 
Since $I$ is locally-serializable, to respect our sequential implementation
of $(LL,\ms{set})$, both operations should complete the write to \emph{head}
before returning. 
Let $\pi_1=\lit{insert}(1)$ be the first operation to write to
\emph{head} in this extended execution. 
Let $\pi_2=\lit{insert}(2)$ be the other insert operation.
It is clear that $\pi_1$ returns $\true$ even though $\pi_2$ overwrites the update of $\pi_1$ on $\ms{head}$
and also returns $\true$. 
Recall that implementations in $\mathcal{P}$ are deadlock-free. 
Thus,   
we can further extend the execution with a complete $\lit{contains}(1)$
that will return $\false$ (the element inserted to the list by $\pi_1$
is lost)---a contradiction since $I$ is linearizable with respect to \emph{set}. 
Thus, $\sigma\notin\S(I)$ for any $I\in \mathcal{P}$.

On the other hand, the schedule $\sigma$ is accepted by $I^{C}\in \mathcal{SM}$, since
there is no conflict between the two concurrent update operations.
\end{proof}
%
%


\section{Relaxed optimistic implementation: proof of correctness}
\label{app:lltm}
Let $\alpha$ be an execution of $I^{RM}$ and $<_\alpha$
denote the total-order on events in $\alpha$.
For simplicity, we assume that $\alpha$ starts with an artificial
sequential execution of an insert operation $\pi_0$ that inserts $\ms{tail}$ and sets $\ms{head}.\ms{next}=\ms{tail}$. 
Let $H$ be the history exported by $\alpha$, where 
all reads and writes are sequential. 
We construct $H$ by associating a linearization point $\ell_{op}$ with each
non-aborted read or write operation $op$ performed in $\alpha$  
as follows:
\begin{itemize}
\item  if $op$ is a read, then performed by process $p_k$,
  $\ell_{op}$ is the base-object $\lit{read}$ in line~\ref{line:linr};
\item  if $op$ is a write within an $\lit{insert}$ operation,
  $\ell_{op}$ is the base-object $\lit{cas}$ in line~\ref{line:linw};
\item  if $op$ is a write  within a $\lit{remove}$ operation,
  $\ell_{op}$ is the base-object  $\lit{cas}$ in line~\ref{line:inswrite}.
\end{itemize}
We say that a $\lit{read}$ of an element $X$ within an operation $\pi$
is \emph{valid} in $H$ 
(we also say that $X$ is \emph{valid}) if 
there does not exist any $\lit{remove}$ operation $\pi_1$ that \emph{deallocates} $X$ (removes $X$ from the list)
such that $\ell_{\pi_1.\lit{write}(X)} <_{\alpha}
\ell_{\pi.\lit{read}(X)}$. 
\begin{lemma}
\label{lem:invar1}
Let $\pi$ be any operation performing $\lit{read}(X)$ followed by $\lit{read}(Y)$ in $H$. 
Then (1) there exists an $\lit{insert}$ operation that sets $X.\ms{next}=Y$ prior to $\pi.\lit{read}(X)$, and
(2) $\pi.\lit{read}(X)$ and $\pi.\lit{read}(Y)$ are \emph{valid} in $H$.
\end{lemma}
\begin{proof} 
Let $\pi$ be any operation in $I^{RM}$ that performs $\lit{read}(X)$ followed by $\lit{read}(Y)$.
If $X$ and $Y$ are \emph{head} and \emph{tail} respectively, $\ms{head}.\ms{next}=\ms{tail}$ (by assumption). Since no $\lit{remove}$ operation
deallocates the \emph{head} or \emph{tail}, the $\lit{read}$ of $X$ and $Y$ are \emph{valid} in $H$.

Now, let $X$ be the \emph{head} element and suppose that $\pi$ performs $\lit{read}(X)$ followed by $\lit{read}(Y)$; $Y\neq \ms{tail}$ in $H$.
Clearly, if $\pi$ performs a $\lit{read}(Y)$, there exists an
operation $\pi'=\lit{insert}$ that has previously set $\emph{head}.\ms{next}=Y$.
More specifically, $\pi.\lit{read}(X)$ performs the action in line~\ref{line:linr} after the
write to shared memory by $\pi'$ in line~\ref{line:inscommit}. 
By the assignment of linearization points to tx-operations, $\ell_{\pi'} <_{\alpha} \ell_{\pi.\lit{read}(X)}$.
Thus, there exists an $\lit{insert}$ operation that sets $X.\ms{next}=Y$ prior to $\pi.\lit{read}(X)$ in $H$.

For the second claim, we need to prove that the $\lit{read}(Y)$ by $\pi$ is \emph{valid} in $H$.
Suppose by contradiction that $Y$ has been deallocated by some $\pi''=\lit{remove}$ operation prior to $\lit{read}(Y)$ by $\pi$.
By the rules for linearization of read and write operations, the action in line~\ref{line:commit-update} precedes the action in line~\ref{line:linr}.
However, $\pi$ proceeds to perform the check in line~\ref{line:flagcheck} and returns $\bot$ since the flag corresponding to the element $Y$ is previously set by $\pi''$. Thus, $H$ does not contain $\pi.\lit{read}(Y)$---contradiction.

Inductively, by the above arguments, every non-\emph{head} $\lit{read}$ by $\pi$ 
is performed on an element previously created by an $\lit{insert}$ operation
and is valid in $H$.
\end{proof} 
\begin{lemma}
\label{lem:rls}
$H$ is locally serializable with respect to $\LL$.
\end{lemma}
\begin{proof}
By Lemma~\ref{lem:invar1}, every element $X$ read 
within an operation $\pi$ 
is previously created by an $\lit{insert}$ operation and is valid in $H$.
Moreover, if the read operation on $X$ returns $v'$, then 
$X.\textit{next}$ stores a pointer to another valid element that
stores an integer value $v''>v'$.
Note that the series of reads performed by $\pi$ terminates as soon as 
an element storing value $v$ or higher is found. Thus, $\pi$ performs at most
$O(|v-v_0|)$ reads, where $v_0$ is the value of the second element read by $\pi$.  
Now we construct $S^{\pi}$ as a sequence of $\lit{insert}$ operations,
that insert values read by $\pi$, one by one, followed by $\pi$. 
By construction, $S^{\pi}\in\Sigma_{\ms{LL}}$.
\end{proof}
%
%
It is sufficient for us to prove that every finite high-level history $H$
of $I^{RM}$ is linearizable.
%
First, we obtain a completion $\tilde H$ of $H$ as follows.
The invocation of an incomplete \lit{contains} operation is discarded.
The invocation of an incomplete $\pi=\lit{insert} \vee \lit{remove}$
operation that has not returned successfully from the $\lit{write}$
operation is discarded; otherwise, it is completed with response $\true$.

We obtain a sequential high-level history $\tilde S$ equivalent to $\tilde H$ by associating a linearization point $\ell_{\pi}$ 
with each operation $\pi$ as follows.
For each $\pi=\lit{insert}\vee\lit{remove}$ that returns $\true$ in $\tilde H$, $\ell_{\pi}$ is associated with 
the first $\lit{write}$ performed by $\pi$ in $H$; otherwise
$\ell_{\pi}$ is associated with the last $\lit{read}$ performed by $\pi$ in $H$. For $\pi=\lit{contains}$ that returns $\true$, $\ell_{\pi}$
is associated with the last $\lit{read}$ performed in $I^{RM}$; otherwise $\ell_{\pi}$ is associated with the $\lit{read}$ of \emph{head}.
Since linearization points are chosen within the intervals of 
operations of $I^{RM}$, for any two operations
$\pi_i$ and $\pi_j$ in ${\tilde H}$, if $\pi_i \rightarrow_{\tilde H}
\pi_j$, then $\pi_i \rightarrow_{\tilde S} \pi_j$.
\begin{lemma}
\label{lem:rlegal}
$\tilde S$ is consistent with the sequential specification of type \textit{set}.
\end{lemma}
\begin{proof}
Let ${\tilde S}^k$ be the prefix of $\tilde S$ consisting of
the first $k$ complete operations. 
We associate each ${\tilde S}^k$ with a set $q^k$ of objects that were
successfully inserted and not subsequently successfully removed in ${\tilde S}^k$.
We show by induction on $k$ that the sequence of state transitions in
${\tilde S}^k$ is consistent with operations' responses in ${\tilde
  S}^k$ with respect to the \textit{set} type. 

The base case $k=1$ is trivial: the \textit{tail} element containing
$+\infty$ is successfully inserted.
Suppose that ${\tilde S}^k$ is consistent with the \emph{set} type and
let $\pi_1$ with argument $v\in \mathbb{Z}$ and response $r_{\pi_{1}}$
be the last operation of ${\tilde S}^{k+1}$.  
We want to show that $(q^k,\pi_1,q^{k+1},r_{\pi_{1}})$ is consistent with the \textit{set} type. 
%
%
\begin{enumerate}
\item[(1)]
If $\pi_1=\lit{insert}(v)$ returns $\true$ in ${\tilde S}^{k+1}$, there does not exist any other $\pi_2=\lit{insert}(v)$ that returns $\true$ in ${\tilde S}^{k+1}$ such that there does not exist any $\lit{remove}(v)$ that returns $\true$; $\pi_2 \rightarrow_{{\tilde S}^{k+1}} \lit{remove}(v) \rightarrow_{{\tilde S}^{k+1}} \pi_1$.
Suppose by contradiction that such a $\pi_1$ and $\pi_2$ exist.
Every successful $\lit{insert}(v)$ operation performs its penultimate $\lit{read}$ on an element $X$ that stores a value $v'<v$ and the last read is performed on an element that stores a value $v''>v$. Clearly, $\pi_1$ also performs a $\lit{write}$ on $X$.
By construction of $\tilde S$, $\pi_1$ is linearized at the release of the \emph{cas} lock on element $X$.
Observe that $\pi_2$ must also perform a $\lit{write}$ to the element $X$ (otherwise one of $\pi_1$ or $\pi_2$ would return $\false$).
By assumption, the write to $X$ in shared-memory by $\pi_2$
(line~\ref{line:inscommit}) precedes the corresponding write to $X$ in
shared-memory by $\pi_2$. If $\ell_{\pi_2} <_{\alpha}
\ell_{\pi_{1}.\lit{read}(X)}$, then $\pi_1$ cannot return $\true$---a contradiction.
Otherwise, if $\ell_{\pi_{1}.\lit{read}(X)}  <_{\alpha} \ell_{\pi_2}$,
then $\pi_1$ reaches line~\ref{line:linw} and return
$\bot$. This is because either $\pi_1$ attempts to acquire the
\emph{cas} lock on $X$ while it is still held by $\pi_2$ or the value
of $X$ contained in the \emph{rbuf} of the process executing $\pi_1$
has changed---a contradiction. 

If $\pi_1=\lit{insert}(v)$ returns $\false$ in ${\tilde S}^{k+1}$, there exists a $\pi_2=\lit{insert}(v)$ that returns $\true$ in ${\tilde S}^{k+1}$ such that there does not exist any $\pi_3=\lit{remove}(v)$ that returns $\true$; $\pi_2 \rightarrow_{{\tilde S}^{k+1}} \pi_3 \rightarrow_{{\tilde S}^{k+1}} \pi_1$. 
Suppose that such a $\pi_2$ does not exist. Thus, $\pi_1$ must perform
its last $\lit{read}$ on an element that stores value $v''>v$, perform
the action in Line~\ref{line:inscommit} and return $\true$---a contradiction.

It is easy to verify that the conjunction of the above two claims prove that $\forall q\in Q$; $\forall v\in \mathbb{Z}$, ${\tilde S}^{k+1}$ satisfies $(q,\lit{insert}(v),q \cup \{v\},(v \not\in q))$.
\item[(2)]
If $\pi_1=\lit{remove}(v)$, similar arguments as applied to $\lit{insert}(v)$ prove that $\forall q\in Q$; $\forall v\in \mathbb{Z}$, ${\tilde S}^{k+1}$ satisfies $(q,\lit{remove}(v),q \setminus \{v\},(v \in q))$.

\item[(3)]
If $\pi_1=\lit{contains}(v)$ returns $\true$ in ${\tilde S}^{k+1}$, there exists $\pi_2=\lit{insert}(v)$ that returns \emph{true} in ${\tilde S}^{k+1}$ such that there does not exist any $\lit{remove}(v)$ that returns \emph{true} in ${\tilde S}^{k+1}$ such that $\pi_2 \rightarrow_{{\tilde S}^{k+1}} \lit{remove}(v) \rightarrow_{{\tilde S}^{k+1}} \pi_1$.
The proof of this claim immediately follows from Lemma~\ref{lem:invar1}.

Now, if $\pi_1=\lit{contains}(v)$ returns $\false$ in ${\tilde S}^{k+1}$, 
there does not exist an $\pi_2=\lit{insert}(v)$ that returns \emph{true} such that 
there does not exist any $\lit{remove}(v)$ that returns \emph{true}; 
$\pi_2 \rightarrow_{{\tilde S}^{k+1}} \lit{remove}(v) \rightarrow_{{\tilde S}^{k+1}} \lit{contains}(v)$.
Suppose by contradiction that such a $\pi_1$ and $\pi_2$ exist. 
Thus, the action in line~\ref{line:inscommit} by the $\lit{insert}(v)$ operation that updates some element, 
say $X$ precedes the action in line~\ref{line:linr} by $\lit{contains}(v)$ that is associated with its 
first $\lit{read}$ (the \emph{head}).
We claim that $\lit{contains}(v)$ must read the element $X'$ newly created by
$\lit{insert}(v)$ and return $\true$---a contradiction to the initial assumption that it returns $\false$.
The only case when this can happen is if there exists a $\lit{remove}$
operation that forces $X'$ to be unreachable from \emph{head} i.e. concurrent to the $\lit{write}$
to $X$ by $\lit{insert}$, there exists a $\lit{remove}$ that sets $X''.\textit{next}$ to $X.\textit{next}$
after the action in line~\ref{line:inswrite} by $\lit{insert}$.
But this is not possible since the \emph{cas} on $X$ performed by the $\lit{remove}$
would return $\false$.
\end{enumerate}
Thus, inductively, the sequence of state transitions in ${\tilde S}$
satisfies the sequential specification of the \textit{set} type. 
\end{proof}
Lemmas~\ref{lem:rls} and~\ref{lem:rlegal} imply:
\begin{theorem}
\label{th:lr}
$I^{RM}$ is \LS-linearizable with respect to $(\LL,\ms{set})$.
\end{theorem}

\section{Complementary experiments}\label{sec:expe}
To confirm the practicality of our highly concurrent optimistic list-based set algorithm we compared its performance against the state-of-the-heart list-based set synchronized with hand-over-hand locking (or lock coupling).
To this end, we implemented the pseudocode of Algorithm~\ref{alg:elastic} in Java without further optimizations and compared it against the Java code from Herlihy and 
Shavit of the hand-over-hand lock-based linked list~\cite{HS08-book}. 

Figure~\ref{fig:alg2-vs-hoh} gives the throughput as the number of operations per millisecond by having from 1 to 64 threads running between 5\% and 
20\% of updates (either remove or insert with same probability) and the rest of contains operations. The list is initially populated with 512 values that are integers taken from 0 to 1024 with uniform distribution. The machine has 
2 8-core Intel Xeon E5-2450 running at 2.1GHz (32-way as each code is hyperthreaded).
Java is 1.7.0\_55 and the JVM is the OpenJDK 64-Bit Server VM. Each point of the graph results from the average of 10 runs of 5 seconds plus 5 seconds to warmup the JVM.
\begin{figure}
  \begin{center}
    \includegraphics[scale=0.9]{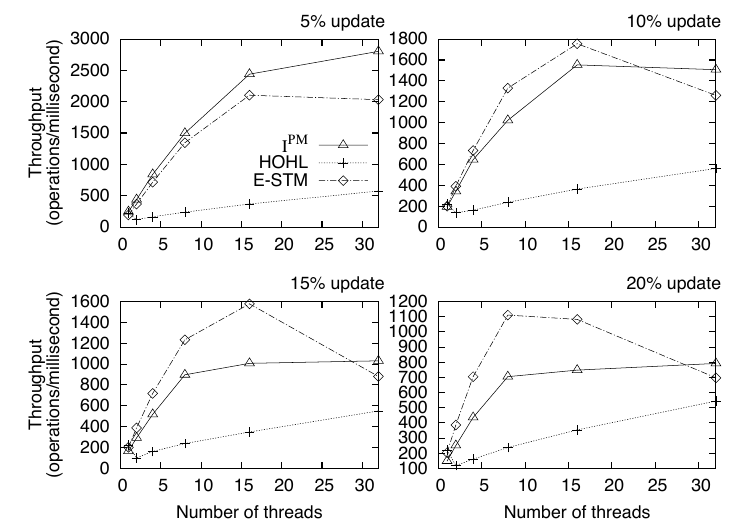}
    \caption{Performance of the list-based set synchronized with our Algorithm~\ref{alg:elastic} ($I^{PM}$), with hand-over-hand locking (HOHL) and with elastic transactions (E-STM)\label{fig:alg2-vs-hoh}}
  \end{center}
\end{figure}

We can observe that our optimistic algorithm ($I^{RM}$) outperforms, in most cases, the list based set synchronized with hand-over-hand locking (HOHL). 
This performance is due to the optimal concurrency of our algorithm (all correct schedules are accepted as shown by Theorem~\ref{th:lrelaxed}) and the low overhead of our algorithm (a valid schedule is efficiently identified and a constant number of $\lit{cas}$ are needed, only during update operations).
We also observe that the peak performance of E-STM is better than the peak performance of $I^{RM}$ after 10\% updates.
This is due to the fact that our code, kept intentionally simple, 
is not optimized with partial aborts, hence it always restarts from the beginning of the list upon abort. By contrast, E-STM is optimized to re-read 
only one node upon some conflict detection~\cite{FGG09}.   The optimal concurrency of our algorithm makes it scale despite contention whereas E-STM does not scale even starting at 5\% udpates.
%
Worth noting is that there is a large body of work on concurrent list-based set algorithms, and we are not claiming our algorithm to be the most efficient. 
Algorithms that are not optimal with respect to concurrency can achieve better results on some workload with a lower overhead. 
An interesting question is how far can our implementation be optimized. In particular, we know that partial aborts and wrappers inlinining could 
boost the performance of our algorithm while retaining its concurrency optimality.

%% file: sequential-fix.tex
\begin{algorithm*}[t]
\caption{Sequential implementation {\LL} (\textit{sorted linked list}) of \emph{set} type}
\label{alg:lists}
  \begin{algorithmic}[1]
  	\begin{multicols}{2}
  	{\size
	
	\Part{Shared variables}{
		\State Initially $\ms{head}$, $\ms{tail}$,
		\State ~~~$\ms{head}.val=-\infty$, $\ms{tail}.val=+\infty$
		\State ~~~$\ms{head}.next=\ms{tail}$
	}\EndPart
	
	\Statex
	
	\Part{$\lit{insert}(v$)}{
                \State $\ms{old} \gets \ms{head}$  		 			\Comment{copy the address}
		\State $\ms{prev} \gets \lit{read}(\ms{head})$  			\Comment{fetch the head node struct}
		\State $\ms{curr} \gets \lit{read}(\ms{prev.next})$ 		\Comment{next element is stored}
		\While{$\ms{curr.val} < v $}
                       \State $\ms{old} \gets \ms{prev.next}$  \Comment{pointer to the next element}
			\State $\ms{prev} \gets \ms{curr}$ 				\Comment{move on}
			\State $\ms{curr} \gets \lit{read}(\ms{curr.next})$ 	\Comment{fetch from memory}
		\EndWhile
		\If{$\ms{curr.val} \neq v$}							\Comment{val is stored locally}
			\State $X \gets \lit{new-node}(v,\ms{prev.next})$ 	\Comment{v and addr. of curr}
			\State $\lit{write}(\ms{old}, [\ms{prev.val},X])$ 	\Comment{stores val and next field}
		\EndIf
		\Return $(\ms{curr.val}\neq v)$ 						
		\EndReturn
   	}\EndPart
	
	\newpage
	
	\Part{$\lit{remove}(v$)}{

	        \State $\ms{old} \gets \ms{head}$  					
		\State $\ms{prev} \gets \lit{read}(\ms{head})$  			\Comment{fetch the head node struct}
		\State $\ms{curr} \gets \lit{read}(\ms{prev.next})$ 		\Comment{next field is stored locally}
		\While{$\ms{curr.val} < v $} 	
					\Comment{the val field is stored locally}
                       \State $\ms{old} \gets \ms{prev.next}$
                       \Comment{pointer to the next element}
			\State $\ms{prev} \gets \ms{curr}$ 				\Comment{move on}
			\State $\ms{curr} \gets \lit{read}(\ms{curr.next})$ 	\Comment{fetch from memory}
		\EndWhile
		\If{$\ms{curr.val} = v$}							\Comment{val is stored locally}
			\State $\lit{write}(\ms{old},
                               [\ms{prev.val},\ms{curr.next}])$
		\EndIf
		\Return $(\ms{curr.val}=v)$ 						
		\EndReturn	
   	}\EndPart
	
	\Statex
	
	\Part{$\lit{contains}(v$)}{
		\State $\ms{prev} \gets \lit{read}(\ms{head})$  			\Comment{fetch the head node struct}
		\State $\ms{curr} \gets \lit{read}(\ms{prev.next})$ 		\Comment{next field is stored locally}
		\While{$\ms{curr.val} < v $} 						\Comment{the val field is stored locally}
			\State $\ms{prev} \gets \ms{curr}$ 				\Comment{move on}
			\State $\ms{curr} \gets \lit{read}(\ms{curr.next})$ 	\Comment{fetch from memory}
		\EndWhile
 	   	\Return $(\ms{curr.val}=v)	$						\Comment{val is stored locally}
 	   	\EndReturn
   	 }\EndPart

	}
	\end{multicols}
  \end{algorithmic}
\end{algorithm*}

%% file: paper.bbl
\def\noopsort#1{} \def\No{\kern-.25em\lower.2ex\hbox{\char'27}}
  \def\no#1{\relax} \def\http#1{{\\{\small\tt
  http://www-litp.ibp.fr:80/{$\sim$}#1}}}
\begin{thebibliography}{10}

\bibitem{AMT10}
Y.~Afek, A.~Morrison, and M.~Tzafrir.
\newblock View transactions: Transactional model with relaxed consistency
  checks.
\newblock In {\em PODC}, 2010.

\bibitem{AFHHT07}
M.~K. Aguilera, S.~Fr{\o}lund, V.~Hadzilacos, S.~L. Horn, and S.~Toueg.
\newblock Abortable and query-abortable objects and their efficient
  implementation.
\newblock In {\em PODC}, pages 23--32, 2007.

\bibitem{AGK11-popl}
H.~Attiya, R.~Guerraoui, D.~Hendler, P.~Kuznetsov, M.~Michael, and M.~Vechev.
\newblock Laws of order: Expensive synchronization in concurrent algorithms
  cannot be eliminated.
\newblock In {\em POPL}, pages 487--498, 2011.

\bibitem{ARR10}
H.~Attiya, G.~Ramalingam, and N.~Rinetzky.
\newblock Sequential verification of serializability.
\newblock In {\em Proceedings of the 37th annual ACM SIGPLAN-SIGACT symposium
  on Principles of programming languages}, pages 31--42, 2010.

\bibitem{AW04}
H.~Attiya and J.~Welch.
\newblock {\em Distributed Computing. Fundamentals, Simulations, and Advanced
  Topics.}
\newblock John Wiley \& Sons, 2004.

\bibitem{BS88}
R.~Bayer and M.~Schkolnick.
\newblock Concurrency of operations on {B}-trees.
\newblock In {\em Readings in database systems}, pages 129--139. Morgan
  Kaufmann Publishers Inc., 1988.

\bibitem{norec}
L.~Dalessandro, M.~F. Spear, and M.~L. Scott.
\newblock {NOrec}: streamlining {STM} by abolishing ownership records.
\newblock In {\em PPOPP}, pages 67--78, 2010.

\bibitem{TMS09}
S.~Doherty, L.~Groves, V.~Luchangco, and M.~Moir.
\newblock Towards formally specifying and verifying transactional memory.
\newblock {\em Electron. Notes Theor. Comput. Sci.}, 259:245--261, Dec. 2009.

\bibitem{FFR08}
P.~Felber, C.~Fetzer, and T.~Riegel.
\newblock Dynamic performance tuning of word-based software transactional
  memory.
\newblock In {\em PPoPP}, pages 237--246, 2008.

\bibitem{FGG09}
P.~Felber, V.~Gramoli, and R.~Guerraoui.
\newblock Elastic transactions.
\newblock In {\em DISC}, pages 93--107, 2009.

\bibitem{GG14}
V.~Gramoli and R.~Guerraoui.
\newblock Democratizing transactional programming.
\newblock {\em Commun. ACM}, 57(1):86--93, Jan 2014.

\bibitem{GHF10}
V.~Gramoli, D.~Harmanci, and P.~Felber.
\newblock On the input acceptance of transactional memory.
\newblock {\em Parallel Processing Letters}, 20(1):31--50, 2010.

\bibitem{GK09-progressiveness}
R.~Guerraoui and M.~Kapalka.
\newblock The semantics of progress in lock-based transactional memory.
\newblock In {\em POPL}, pages 404--415, 2009.

\bibitem{tm-theory}
R.~Guerraoui and M.~Kapalka.
\newblock Transactional memory: Glimmer of a theory.
\newblock In {\em Proceedings of the 21st International Conference on Computer
  Aided Verification}, CAV '09, pages 1--15, Berlin, Heidelberg, 2009.
  Springer-Verlag.

\bibitem{tm-book}
R.~Guerraoui and M.~Kapalka.
\newblock {\em Principles of Transactional Memory, Synthesis Lectures on
  Distributed Computing Theory}.
\newblock Morgan and Claypool, 2010.

\bibitem{HLR10}
T.~Harris, J.~R. Larus, and R.~Rajwar.
\newblock {\em Transactional Memory, 2nd edition}.
\newblock Synthesis Lectures on Computer Architecture. Morgan {\&} Claypool
  Publishers, 2010.

\bibitem{harris-set}
T.~L. Harris.
\newblock A pragmatic implementation of non-blocking linked-lists.
\newblock In {\em DISC}, pages 300--314, 2001.

\bibitem{HHL+05}
S.~Heller, M.~Herlihy, V.~Luchangco, M.~Moir, W.~N. Scherer, and N.~Shavit.
\newblock A lazy concurrent list-based set algorithm.
\newblock In {\em OPODIS}, pages 3--16, 2006.

\bibitem{Her90}
M.~Herlihy.
\newblock Apologizing versus asking permission: optimistic concurrency control
  for abstract data types.
\newblock {\em ACM Trans. Database Syst.}, 15(1):96--124, 1990.

\bibitem{Her91}
M.~Herlihy.
\newblock Wait-free synchronization.
\newblock {\em ACM Trans. Prog. Lang. Syst.}, 13(1):123--149, 1991.

\bibitem{HM93}
M.~Herlihy and J.~E.~B. Moss.
\newblock Transactional memory: architectural support for lock-free data
  structures.
\newblock In {\em ISCA}, pages 289--300, 1993.

\bibitem{HS08-book}
M.~Herlihy and N.~Shavit.
\newblock {\em The art of multiprocessor programming}.
\newblock Morgan Kaufmann, 2008.

\bibitem{HS11-progress}
M.~Herlihy and N.~Shavit.
\newblock On the nature of progress.
\newblock In {\em OPODIS}, pages 313--328, 2011.

\bibitem{HW90}
M.~Herlihy and J.~M. Wing.
\newblock Linearizability: A correctness condition for concurrent objects.
\newblock {\em ACM Trans. Program. Lang. Syst.}, 12(3):463--492, 1990.

\bibitem{damien-vw}
D.~Imbs, J.~R.~G. de~Mend\'{\i}vil, and M.~Raynal.
\newblock Brief announcement: virtual world consistency: a new condition for
  stm systems.
\newblock In {\em PODC}, pages 280--281, 2009.

\bibitem{KP79}
H.~T. Kung and C.~H. Papadimitriou.
\newblock An optimality theory of concurrency control for databases.
\newblock In {\em SIGMOD}, pages 116--126, 1979.

\bibitem{KR11}
P.~Kuznetsov and S.~Ravi.
\newblock On the cost of concurrency in transactional memory.
\newblock In {\em OPODIS}, pages 112--127, 2011.
\newblock full version: http://arxiv.org/abs/1103.1302.

\bibitem{Pap79-serial}
C.~H. Papadimitriou.
\newblock The serializability of concurrent database updates.
\newblock {\em J. ACM}, 26:631--653, 1979.

\bibitem{ST95}
N.~Shavit and D.~Touitou.
\newblock Software transactional memory.
\newblock In {\em PODC}, pages 204--213, 1995.

\bibitem{vechev08}
M.~T. Vechev and E.~Yahav.
\newblock Deriving linearizable fine-grained concurrent objects.
\newblock In {\em PLDI}, pages 125--135, 2008.

\bibitem{Wei88}
W.~E. Weihl.
\newblock Commutativity-based concurrency control for abstract data types.
\newblock {\em IEEE Trans. Comput.}, 37(12):1488--1505, 1988.

\bibitem{Wei86}
G.~Weikum.
\newblock A theoretical foundation of multi-level concurrency control.
\newblock In {\em PODS}, pages 31--43, 1986.

\bibitem{WV02-book}
G.~Weikum and G.~Vossen.
\newblock {\em Transactional Information Systems: Theory, Algorithms, and the
  Practice of Concurrency Control and Recovery}.
\newblock Morgan Kaufmann, 2002.

\bibitem{Yan84}
M.~Yannakakis.
\newblock Serializability by locking.
\newblock {\em J. ACM}, 31(2):227--244, 1984.

\end{thebibliography}
